\documentclass[sn-mathphys,Numbered]{sn-jnl}

\usepackage{graphicx}%
\usepackage{multirow}%
\usepackage{amsmath,amssymb,amsfonts}%
\usepackage{amsthm}%
\usepackage{mathrsfs}%
\usepackage[title]{appendix}%
\usepackage{xcolor}%
\usepackage{textcomp}%
\usepackage{manyfoot}%
\usepackage{booktabs}%
\usepackage{algorithm}%
\usepackage{algorithmicx}%
\usepackage{algpseudocode}%
\usepackage{listings}%

\theoremstyle{thmstyleone}%
\newtheorem{theorem}{Theorem}

\newtheorem{lemma}{Lemma}%
\newtheorem{proposition}[theorem]{Proposition}%
\newtheorem{notation}{Notation}%
\theoremstyle{thmstyletwo}%
\newtheorem{example}{Example}%
\newtheorem{remark}{Remark}%

\theoremstyle{thmstylethree}%
\newtheorem{definition}{Definition}%

\begin{document}

\title[Minimum distances of binary optimal LCD codes of dimension five are completely determined]{Minimum distances of binary optimal LCD codes of dimension five are completely determined}


\author*[1,2]{\fnm{Yang} \sur{Liu}}\email{liu\_yang10@163.com}

\author[1]{\fnm{Ruihu} \sur{Li}}\email{liruihu@aliyun.com}
\equalcont{These authors contributed equally to this work.}

\author[1]{\fnm{Qiang} \sur{Fu}}\email{fuqiangkgd@163.com}
\equalcont{These authors contributed equally to this work.}

\author[1]{\fnm{Hao} \sur{Song}}\email{song\_hao@163.com}
\equalcont{These authors contributed equally to this work.}

\affil*[1]{\orgdiv{Department of Basic Sciences}, \orgname{Air Force
Engineering  University}, \orgaddress{ \city{Xi'an},
\postcode{710051}, \state{Shannxi}, \country{China}}}

\affil[2]{\orgdiv{ Air Defense and antimissile school}, \orgname{Air
Force Engineering  University}, \orgaddress{ \city{Xi'an},
\postcode{710051}, \state{Shannxi}, \country{China}}}


\abstract{Let $d_{a}(n,5)$ and $d_{l}(n,5)$ be the minimum weights
of binary [n,5] optimal linear codes and linear complementary dual
(LCD) codes, respectively. This article aims to investigate
$d_{l}(n,5)$ of some families of binary [n,5] LCD  codes when
$n=31s+t\geq 14$ with $s$ an integer and $t \in
\{2,8,10,12,14,16,18\}$. By determining the defining vectors of
optimal linear codes and discussing their reduced codes, we classify
optimal linear codes and calculate their hull dimensions. Thus, the
non-existence of these classes  of binary $[n,5,d_{a}(n,5)]$ LCD
codes are  verified and we further derive that
$d_{l}(n,5)=d_{a}(n,5)-1$ for $t\neq 16$ and
$d_{l}(n,5)=16s+6=d_{a}(n,5)-2$ for $t=16$. Combining with known
results on optimal LCD code,  $d_{l}(n,5)$ of all $[n,5]$ LCD codes
are completely determined.}

\keywords{optimal code, LCD code, hull dimension, defining vector,
reduced code}



\maketitle

\section{Introduction}

Let $F_{2}^{n}$ be the $n$-dimensional row vector space over binary
field $F_{2}$. A binary linear $[n,k]$ code is a $k$-dimensional
subspace of $F_{2}^{n}$. The weight $w(x)$ of a vector $x\in
F_{2}^{n}$ is the number of its nonzero coordinates. If the minimum
weight of nonzero  vectors in $\mathcal{C}$$=[n,k]$ is $d$, then $d$
is called the minimum distance of $\mathcal{C}$ and the code
$\mathcal{C}$ is denoted as $[n,k,d]$. A linear code
$\mathcal{C}$$=[n,k,d]$ is optimal if its minimum distance $d$ can
meet the largest value for given $n,k$,  which is denoted as
$\mathcal{C}$$=[n,k,d_{a}(n,k)]=[n,k,d_{a}]$.
 Two binary codes $\mathcal{C}$ and $\mathcal{C}'$ are
equivalent  if one can be obtained from the other by permuting the
coordinates [1], they are denoted  as  $\mathcal{C}$ $ \cong$
$\mathcal{C}'$. A matrix whose rows form a basis of $\mathcal{C}$ is
called a generator matrix of this code.

 The dual code $\mathcal{C}^{\perp}$ of $\mathcal{C}$ is defined as
$\mathcal{C}$$^{\perp}$= $\{x\in F_{2}^{n} \mid x\cdot y=xy^{T}=0
 \hbox{ for all}\  y \in  \mathcal{C} \}$.
 A code $\mathcal{C}$ is {\it self-orthogonal} (SO) if
 $\mathcal{C}\subseteq\mathcal{C}^{\perp}$.
  The hull of a linear code $\mathcal{C}$ was defined as
$Hu(\mathcal{C})=\mathcal{C}^{\perp}\cap\mathcal{C}$  in [2], and
was called a radical code of $\mathcal{C}$ in the nomenclature of
classical group in [3]. Define $h(\mathcal{C})=$dim$Hu(\mathcal{C})$
as the {\it hull dimension} of $\mathcal{C}$ and
$h([n,k,d])=min\{h(\mathcal{C})\mid \mathcal{C} \hbox{ is a binary
$[n,k,d] $ code } \}$.

If $Hu(\mathcal{C})=\{0\}$ (or $h(\mathcal{C})=0$), $\mathcal{C}$ is
an LCD code [4]. LCD cyclic codes were introduced by Massey [4] and
gave an optimal linear coding solution for the two user binary adder
channel. Carlet et al. showed that LCD codes can be used  to fight
against side-channel attacks [5]. In recent years, much work has
been done on property and construction of LCD codes [5-15,24,25]. It
has been shown in [6] that any code over $F_{q}$ is equivalent to
some LCD code for $q \geq 4$,  which motivates people to study
binary and ternary LCD codes. In this paper we focus on the hull
dimension of binary optimal codes and  LCD codes.

It is an important problem to determine the largest minimum weight
$d_{l}(n,k)$ among all LCD $[n, k]$ codes and to construct  LCD $[n,
k,d_{l}(n,k)]$ codes for given $n, k$  [5-15,24,25]. Recently,
construction of optimal LCD  codes with short lengths or low
dimension are discussed, and low and upper bound for $d_{l}(n,k)$
have been established  in  [6-14]. If $n\leq 24$ and $1\leq k\leq
n$,
 $d_{l}(n,k)$ were determined.  If $k\leq n\leq 40$,
most of $d_{l}(n,k)$ were determined in [6-15]. If $k\leq 4$, all
$d_{l}(n,k)$ were determined in [8-12]. As for $k=5$, $d_{l}(n,5)$
were partially determined in [11-13] except $n=31s+t\geq 40$ and $t
\in \{2,8,10,12,14,16,18\}$. In [15], Li {\it  et al.} introduced
the {\it reduced code} of a linear code and developed some new
approach to determine upper bounds on $d_{l}(n,6)$ by determining
hull dimensions of $[n, 6]$ optimal linear codes, and construct many
optimal  $[n, 6]$ LCD codes.

A code $\mathcal{C}$$=[n,k]$ with  generator matrix $G$ is an LCD
code if and only if the matrix $GG^{T}$ is invertible [4]. Thus, to
prove non-existence of an $[n,k, d_{a}]$ LCD code, one only needs to
 verify  $h=k-(rank (GG^{T}))\geq 1$
  for each $\mathcal{C}=[n,k,d_{a}]$ with  generator matrix $G$, that is to show $h([n,k, d_{a}])\geq 1$.


In  [17],  Li {\it  et al.} introduced two concepts called the {\it
defining vector} and {\it weight vector}
 of an $[n, 5,d]$  linear code, and  established relations among parameters
 of this code, its defining vector and weight vector. They
changed the problem of determining linear codes into solving the
system of linear equations. Further research on  defining vectors
and weight vectors of optimal linear codes and their applications
were made in [18,19]. The classification of all $[n, k]$ optimal
linear codes with $k\leq 4$ [18] and some $[n, k]$ optimal linear
codes with $k\geq 5$ were determined [18,19].

 Inspired by Refs. [15,17-19], we will show all $[n, 5]$ optimal
linear codes are not LCD  for $n=31s+t\geq 14$ and $t \in
\{2,8,10,12,14,16,18\}$. Now set $k=5$ and $N=31$. Denote
$L=(l_{1},l_{2},\cdots,l_{31})$ as a  defining  vector of a given
$[n, 5, d_{a}]$ (for details see Section 2), and let
$l_{max}=max_{1\leq i\leq N}\{ l_{i} \}$, $l_{min}=min_{1\leq i\leq
N}\{  l_{i} \}$. The main techniques used in this manuscript can be
briefly described  as follows (for details see Section 3):

 (1) From parameters of $[n,5, d_{a}]$,  estimate $l_{max}$ and $l_{min}$ according to Ref. [18].

 (2) According to values of $l_{max}$ and $l_{min}$, analyze conditions under which a reduced code $\mathcal{D}$ of
 an  $[n,5, d_{a}]$ can satisfy $h(\mathcal{D})\geq 2$.

 (3) If $l_{max}$ and $l_{min}$ do not satisfy (2),  determine all such $L$'s and all $[n,5,
 d_{a}]$ codes with  defining vectors $L$'s, classify  $[n,5, d_{a}]$ codes and calculate
 their  hull dimensions and  their weight enumerators.

  Our main result in this paper is  Theorem \ref{the1.1}.

\begin{theorem}\label{the1.1}
 If $s$ is an integer,  $t\in \{2, 8,10,12,
14,16,18\}$ and $n=31s+t\geq 14$,  then an optimal
$[n,5,d_{a}(n,5)]$ linear code is not an LCD code, and
$d_{l}(n,5)=d_{a}-1$ if $t\neq 16$ and $d_{l}(n,5)=d_{a}-2$ if $t=
16$.
\end{theorem}

  Combining with results of Refs. [11-14] on optimal LCD codes, we can completely
 determine $d_{l}(n, 5)$  for all $n\geq 5$, which is shown in Table \ref{tab1} and  Theorem \ref{the1.2}
.

\begin{table}[h!]
\caption{Minimum distances of optimal LCD $[n,5]_{2}$ codes with
$n=31s+t\ge 14$ }\label{tab1}

\begin{tabular} {l|llllllllllll}
 \hline \hline
  n&$\scriptsize{31s}$  &$\scriptsize{31s+1}$ & $\scriptsize{31s+2}$&$\scriptsize{31s+3}$ &$\scriptsize{31s+4}$&
$\scriptsize{31s+5}$
&$\scriptsize{31s+6}$\\
\hline
 $d_a$&$\scriptsize{16s}$  & $\scriptsize{16s}$ & $\scriptsize{16s}$ & $\scriptsize{16s}$ & $\scriptsize{16s}$ &
$\scriptsize{16s+1}$ & $\scriptsize{16s+2}$  \\
\hline
 $d_l$&$\scriptsize{16s-2}$  & $\scriptsize{16s-1}$ & $\scriptsize{16s-1}$ & $\scriptsize{16s}$ & $\scriptsize{16s}$ &
$\scriptsize{16s+1}$ &$\scriptsize{16s+1}$ \\
 \hline \hline

  n&$\scriptsize{31s+7}$&$\scriptsize{31s+8}$&$\scriptsize{31s+9}$&
\!\!\!$\scriptsize{31s+10}$ &$\scriptsize{31s+11}$&$\scriptsize{31s+12}$&$\scriptsize{31s+13}$\\
\hline
 $d_a$&\!\!\! $\scriptsize{16s+2}$&
$\scriptsize{16s+3}$& $\scriptsize{16s+4}$ & $\scriptsize{16s+4}$&
$\scriptsize{16s+4}$& $\scriptsize{16s+5}$& $\scriptsize{16s+6}$
  \\
\hline
 $d_l$&$\scriptsize{ 16s+2}$&
$\scriptsize{16s+2}$& $\scriptsize{16s+3}$ & $\scriptsize{16s+3}$ &
$\scriptsize{16s+4}$& $\scriptsize{16s+4}$& $\scriptsize{16s+5}$
  \\
 \hline \hline
   n& $\scriptsize{31s+14}$
&$\scriptsize{31s+15}$&$\scriptsize{31s+16}$ &
\!\!\!$\scriptsize{31s+17}$&$\scriptsize{31s+18}$
&$\scriptsize{31s+19}$&\!\!\! $\scriptsize{31s+20}$\\
\hline
  $d_a$ &  $\scriptsize{16s+6}$&
$\scriptsize{16s+7}$&  $\scriptsize{16s+8}$ & $\scriptsize{16s+8}$ &
$\scriptsize{16s+8}$ & $\scriptsize{16s+8}$ & $\scriptsize{16s+9}$
  \\
\hline $d_l$ &  $\scriptsize{16s+5}$& $\scriptsize{16s+6}$&
$\scriptsize{16s+6}$ & $\scriptsize{16s+7}$ & $\scriptsize{16s+7}$ &
$\scriptsize{16s+8}$ &$\scriptsize{16s+9}$\\
 \hline \hline
n&$\scriptsize{31s+21}$&$\scriptsize{31s+22}$&$\scriptsize{31s+23}$&$\scriptsize{31s+24}$&
\!\!\!$\scriptsize{31s+25}$&$\scriptsize{31s+26}$&$\scriptsize{31s+27}$\\
\hline
 $d_a$&   $\scriptsize{16s+10}$&
$\scriptsize{16s+10}$& $\scriptsize{16s+11}$& $\scriptsize{16s+12}$
& $\scriptsize{16s+12}$& $\scriptsize{16s+12}$&
$\scriptsize{16s+13}$
  \\
\hline
 $d_l$&   $\scriptsize{16s+9}$&
$\scriptsize{16s+10}$& $\scriptsize{16s+10}$ &
$\scriptsize{16s+11}$& $\scriptsize{16s+11}$& $\scriptsize{16s+12}$&
$\scriptsize{16s+12}$\\
 \hline \hline

n&$\scriptsize{31s+28}$&$\scriptsize{31s+29}$
&$\scriptsize{31s+30}$\\
\hline
 $d_a$& $\scriptsize{16s+14}$& $\scriptsize{16s+14}$&
$\scriptsize{16s+15}$
  \\
\hline
 $d_l$& $\scriptsize{16s+13}$& $\scriptsize{16s+13}$&
$\scriptsize{16s+14}$\\
 \hline \hline
\end{tabular}
\end{table}

\begin{theorem}\label{the1.2} If $n=31s+t\geq 5$, then there are optimal
 LCD codes as follows:

(1) ([7,8]) If $5\leq n\leq 13$ and $n\neq 6,10$, then there is an
$[n,5,d_{a}(n, 5)]$ optimal  LCD code, while $n= 6,10$, an optimal
LCD $[n,5,d_{a}(n, 5)-1]$ exists.

(2) ([9-13])  If $t=3,4,5,7,11,19,20,22,26$, $n=31s+t\geq 14$, there
is an  $[n,5,d_{a}(n, 5)]$ optimal LCD code.

(3) If $t\neq 0, 3,4,5,7,11,16,19,20,22,26$ and $n=31s+t\geq 14$,
there is an $[n,5,d_{a}(n, 5)-1]$ optimal LCD code according to
Refs. [9-13] and Theorem \ref{the1.1} above.

(4) If $t= 0, 16$ and $n=31s+t\geq 14$, there is an $[n,5,d_{a}(n,
5)-2]$  optimal LCD code according to Ref. [17] and Theorem
\ref{the1.1} above.

\end{theorem}

\begin{remark} From Ref. [16], it is easy to know all optimal [n,5]
linear codes   can achieve the Griesmer bound for $14 \leq n \leq
256$. For $n> 256$,  the length $n$ can be denoted as
 $n=31s+t$,  where $s\geq 7$ and $31\leq t\leq 61$ are integers.
 By the juxtaposition of $s$ simplex codes [31,5,16]
 and an optimal linear code $[t, 5, d_{a}(t,5)]$, one can easily obtain  all $[n, 5, d_{a}(n,5)]$ optimal
 linear codes with $d_{a}(n, 5)$ achieving the Griesmer bound for $
 n>256$. That is to say any $d_{a}(n, 5)$ can be obtain by the Griesmer
 bound for all length $n\geq 14$. It naturally follows that $d_{l}(n, 5)$ can be  denoted by  $d_{a}(n,
 5)$ as Theorems \ref{the1.1} and \ref{the1.2}.
\end{remark}

 The rest of this paper is organized as follows. In Section 2, some
definitions, notations and basic results about optimal LCD codes are
given. The proof of the main result Theorem \ref{the1.1} is provided
in Section 3. Section 4 gives conclusion and discussion.

\section{Preliminaries}

In this section,  some concepts and notations are given for later
use. The all-one vector and  zero vector of length $n$ are defined
as $\bf{1_{n}}$=$(1,1,\cdots,1)_{1\times n}$ and
$\bf{0_{n}}$=$(0,0,...,0)_{1\times n}$, respectively. Let
$iG=(G,G,\cdots,G)$ be the juxtaposition  of $i$ copies of $G$ for
given matrix $G$, then the juxtaposition  of $i$ copies of
$\mathcal{C}=[n,k]$ can be denoted as $i\mathcal{C}$ with generator
matrix $iG$. In this article,  we consider linear codes without zero
coordinates and matrices without zero columns.

We introduce some concepts and results in [17-19] at  first. Let
$N=2^{k}-1$, consider

\begin{center}
$\scriptsize{\mathbf{S}_{2}=\left( \begin{array}{lllllll}
101\\
011\\
\end{array} \right)}$, $\scriptsize{\mathbf{S}_{3} =\left(
\begin{array}{ccccccccccccc}
\mathbf{S}_{2}& \mathbf{0}^{T}_{2}&\mathbf{S}_{2}\\
 \mathbf{0}_{3}&1&\mathbf{1}_{3}\\
\end{array} \right)}, \cdots,$
 $\scriptsize{\mathbf{S}_{k+1}=\left(
\begin{array}{ccccccccccccc}
\mathbf{S}_{k}& \mathbf{0}^{T}_{k}&\mathbf{S}_{k}\\
 \mathbf{0}_{2^{k}-1}&1&\mathbf{1}_{2^{k}-1}\\
\end{array} \right)}$.
\end{center}
The matrix $S_{k}$ generates the $k$-dimensional simplex code
$\mathcal{S}_{k}=[2^{k}-1,k,2^{k-1}] $. Let $\alpha_{i}$ be the
$i$-th column of $S_{k}$ for $1\leq i\leq N$. The last $2^{k}-2^{m}$
columns of $S_{k}$ form a matrix $M_{k,m}$ for $1\leq m\leq k-1$,
$M_{k,m}$ generates the $k$-dimensional
$\mathcal{MD}_{k,m}=[2^{k}-2^{m},k,2^{k-1}-2^{m-1}] $ MacDonald code
[20]. Simplex codes $\mathcal{S}_{k}$ and MacDonald codes
$\mathcal{MD}_{k,m}$  for $k\geq 4$ will be used to discuss the hull
dimensions of some optimal codes.

Let $N=2^{k}-1$ and  $G=G_{k\times n}$ be a generator matrix of
$\mathcal{C}=[n,k]$. If there are $l_{i}$ copies of $\alpha_{i}$ in
$G$ for $1\leq i\leq N$,  we denote $G$ as
$G=(l_{1}\alpha_{1},\cdots,l_{N}\alpha_{N})$ for short, and call
$L=(l_{1},\cdots,l_{N})$ the {\it defining vector}  of $G$ or
$\mathcal{C}$. Let
 $l_{j_{l}}$ ($1\leq l\leq t$) be
different coordinates of $L=(l_{1},l_{2},\cdots,l_{N})$ with
$l_{j_{1}}< l_{j_{2}}<\cdots<l_{j_{t}}$ in ascending order by  the
number of equal $l_{j_{l}}$. If there are $m_{l} $ entries equal to
$l_{j_{l}}$,  we say  $L$ is of {\it  type}
$]](l_{j_{l}})_{m_{l}}\mid \cdots\mid(l_{j_{t}})_{m_{t}}]]$. For
example, a code with defining vector  $L_{1}=(3,1,1,3,1,3,1)$ is an
SO code, this can be derived from  type   $]](1)_{4}\mid (3)_{3}]]$
of $L_{1}$, and $L_{2}=(s+1,s-1,s,s,s+1,s-1,s+1)$ is of type
$]](s-1)_{2}\mid (s)_{2}\mid(s+1)_{3}]]$.

Parameters and some properties of an $[n,k,d]$ code  can be derived
from its defining vector $L$.  Relations among these objects are
connected by some matrices $P_{k}$ and $Q_{k}$ derived from simplex
code $\mathcal{S}_{k}$ [17,18].   On the other hand, if $[n,k,
d_{a}]$ is an optimal code, we can determine all defining vector
$L$'s whose corresponding  codes have such parameters  by solving
linear equations.   We adopt the treatment of Ref. [18] here, which
is equivalent to that of Ref. [17].

 Let $J_{k}$ be the $(2^{k}-1)\times (2^{k}-1)$ all-one  matrix
 and $P_{2}$ be a $(2^{2}-1)\times (2^{2}-1)$ matrix whose
rows are the non-zero codewords of $\mathcal{S}_{2} $. Using
recursive method, construct
$$P_{2}=\left( \begin{array}{ccccccc}
101\\
011\\
110\\
\end{array} \right),
P_{3}= \left(
\begin{array}{ccccccc}
P_{2}&0&P_{2}\\
\bf{0_{3}}&1& \bf{1_{3}}\\
P_{2}&\bf{1^{T}_{3}}&Q_{2}\\
\end{array}
\right), \cdots, P_{k+1}= \left(
\begin{array}{ccccccc}
P_{k}&\bf{0^{T}_{2^{k}-1}}&P_{k}\\
\bf{0_{2^{k}-1}}&1& \bf{1_{2^{k}-1}}\\
P_{k}&\bf{1^{T}_{2^{k}-1}}&Q_{k}\\
\end{array}
\right),$$
 where $Q_{k}=J_{k}-P_{k}$ for $k\geq 2$.
 Then the seven rows of $P_{3}$  are just  the seven nonzero vectors
of the simplex code $\mathcal{S}_{3}$  $=[7, 3, 4] $. For $k\geq 3$,
then the matrix formed by nonzero codewords of $k+1$-dimensional
simplex code can be obtained from $P_{k}$. Each row of $P_{k}$ has
$(2^{k-1})$'s ones and  $(2^{k-1}-1)$' zeros. Hence each row of
$Q_{k}$ has $(2^{k-1}-1)$'s ones and $(2^{k-1})$'s zeros. According
to Ref. [18], $P_{k}$ and $Q_{k}$ are symmetric matrices, and the
matrix $P_{k}$ is invertible over the rational field and
$P^{-1}_{k}=\frac{1}{2^{k-1}}[J_{k}-2Q_{k}]$.

 If  $\mathcal{C}=[n,k]$ has a generator matrix
 $G=(l_{1}\alpha_{1},\cdots,l_{N}\alpha_{N})$,
the distance $d$ of $\mathcal{C}$ and its codewords weight can be
determined by its  defining vector $L=(l_{1},\cdots,l_{N})$. Let
 $$W^{T}=P_{k}L^{T},$$ then $W=(w_{1},w_{2},\cdots,w_{N})$ is  a vector
formed by weights of $2^{k}-1$ nonzero codewords of $\mathcal{C}$
and $d$$=min_{1\leq i\leq 2^{k}-1 }\{w_{i}\}$ is the distance of
$\mathcal{C}$.  $W$ is called the {\it weight vector} of
$\mathcal{C}$ [17-19]. Suppose $$W=d{\bf1}_{2^{k}-1}+\Lambda,$$
where
$$\Lambda=(\lambda_{1},\lambda_{2},\cdots,\lambda_{N})$$ with
$\lambda_{i}=w_{i}-d\geq 0$ and at least one $\lambda_{i}=0$. Denote
$\sigma=\lambda_{1}+\lambda_{2}+\cdots+\lambda_{N}$, then
$$\sigma=2^{k-1}n-d(2^{k}-1)$$ from  $W^{T}$=$P_{k}L^{T}$.

Suppose there is an $[n,k,d]$ code, to determine the defining vector
$L=(l_{1},l_{2},\cdots,l_{N})$, one can solve the system of  linear
equations
$$\qquad \qquad \qquad L^{T}=P^{-1}_{k}W^{T}
=\frac{1}{2^{k-1}}[(d+\sigma){\bf 1}^{T}_{2^{k}-1}-2Q_{k}
\Lambda^{T}].\qquad \qquad \qquad ( \star)$$
  By determining all
nonnegative integer solutions $L$ of the linear equations  for given
$\sigma=2^{k-1}n-d(2^{k}-1)$, one can obtain all $[n,k,d]$ codes and
their weight distributions using software MATLAB [22]. The process
of solving the linear equations were simplified in [17,18], and
uniqueness of some optimal codes were derived as the following known
conclusions.

\begin{proposition} ([17] Theorem 1.1)
 Suppose $k\geq 3$, $s\geq
1$, $1\leq t\leq 2^{k}-2$ and $n=(2^{k}-1)s+t$. Then every  binary
$[n,k,d]$ code with $d\geq (2^{k-1})s$ and without zero coordinates
is equivalent to a code with generator matrix
$G=((s-c(k,s,t))$$\scriptsize{\mathbf{S}_{k}}$$\mid B)$, where
$c(k,s,t)\leq min\{s,t\}$ is a function of $k,s$ and $t$, and $B$
has $(2^{k}-1)c(k,s,t)+t$ columns.
\end{proposition}

\begin{notation}
 For $s\geq 0$, $n=31s+t\geq 14$ with $t
\in \{2,8,10,12,14,16,18\}$, one can check that an
$[n,5,d_{a}(n,5)]$ optimal linear code without zero coordinates is
equivalent to a code with generator matrix
$G=((s-c(k,s,t))$$\scriptsize{\mathbf{S}_{k}}$$\mid B)$, where
$c(k,s,t)\leq  2$  and $B$ has $(2^{k}-1)c(k,s,t)+t$ columns. To
determine all nonnegative integer solutions $L$ of the system of
linear equations for given $\sigma=2^{k-1}n-d(2^{k}-1)$, one only
needs to determine all nonnegative integer solutions for fixed
lengths $n'=(2^{k}-1)c(k,s,t)+t$ (see Section 3 for details).
\end{notation}

\begin{lemma}
  Let $s\geq 1$, $k\geq 4$, $1\leq
m\leq k-1$, $N=2^{k}-1$. Then the following holds:\\
 1 ( [17] Corollary 2.2) Every $[sN, k, s2^{k-1}]$ code is equivalent
to the SO code with generator matrix
$s\scriptsize{\mathbf{S}_{k}}$.\\
2 ([18,19]) Each $[n,k, d_{a}]=[sN+2^{k}-2^{m},
k,s2^{k-1}+2^{k-1}-2^{m-1}]$ code is equivalent to the code
$\mathcal{MD}$$_{s}(k,m)$, the juxtaposition of $s\mathcal{S}$$_{k}$
and
a $\mathcal{MD}$$(k,m)$ code.\\
Hence, if $m=1,2,$ and  $\geq 3$, then $h([n,k, d_{a}])=k-1,k-2$,
$k$, respectively. \end{lemma}

For  some special  $[n,k, d_{a}]$ optimal codes, it has been shown
$h([n,k, d_{a}])$ $\geq 1$ in [15]. And $h([n,k, d_{a}])$ can also
be estimated from  extended codes or low dimension codes. Thus, we
need the following results of Ref. [15].

\begin{definition}Let $G$ be a generator matrix of $\mathcal{C}=[n, k,d]$ and
$G_{1}$ be a generator matrix of $\mathcal{C}_{1}=[n-m, k-1,\geq
d]$. Suppose $u$ is a matrix of $1$ row and $n-m$ columns. Define
$\mathbf{0}_{k-1,m}$ as the zero matrix with $k-1$ rows and $m$
columns. If
$$G=\left(\begin{array}{ccccc}
                   \mathbf{1}_{m}  & u \\
                      \mathbf{0}_{k-1,m} &G_{1}\\
                    \end{array}
                 \right),$$ \\
then $\mathcal{C}_{1}$ is called a reduced code of
$\mathcal{C}$.\end{definition}

\begin{lemma} If $\mathcal{C}_{1}$ is a reduced code of
$\mathcal{C}=[n, k,d]$ and $h(\mathcal{C}_{1})=r\geq 2$, then
  $h(\mathcal{C})\geq r-1\geq 1$ and  $\mathcal{C}$ is not an  LCD
  code.\end{lemma}

\begin{lemma}\label{lem3}If $d$ is odd,  $\mathcal{C}^{e}$ is  an
extended code of $\mathcal{C}=[n, k,d]$ and
$h(\mathcal{C}^{e})=r\geq 2$, then $h(\mathcal{C})\geq r-1\geq 1$
and  $\mathcal{C}$ is not an LCD code.\end{lemma}

\section{The  proof of  Theorem \ref{the1.1}}

In this section, Theorem \ref{the1.1}  will be proved  by showing
$h([31s+t,5, d_{a}])\geq 1$ for $t\in \{2, 8,10,12,14, 16,18\}$ and
$h([31s+t,5, d_{a}-1])\geq 1$ for $t=16$. Our discussions are
presented  in four subsections. The first subsection verifies
$h([31s+t,5, d_{a}])\geq 1$ for $t \in \{2, 8,12,16\}$, while the
other subsections prove $h([31s+t,5, d_{a}])\geq 1$ for $t=10,14$
and $18$, respectively.

\subsection{\bf  $h([32s+2,5, d_{a}])\geq 1$  and  $h([32s+t ,5, d_{a}])\geq 2$ for $t=8, 12, 16$}

\begin{lemma} If $s\geq 1$, a $[31s+2,5,16s]$ code has $h\geq
1$ and a $[31s+9,5,16s+4]$ code has $h\geq 3$,
 hence they are not LCD codes.\end{lemma}

\begin{proof} A $[31s+2,5,16s]$ code has a reduced code
$[30s+1,4,16s]$, this reduced code can  give a reduced code
$[28s,3,16s]=[7\times 4s,3, 4\times 4s]$, which is an SO code. Thus,
a $[31s+2,5,16s]$ code has $h([31s+2,5,16s])\geq 1$.

A $[31s+9,5,16s+4]$ code has  a reduced code
$[30s+8,4,16s+4]$$=[15\times2s+8,4,8\times2s+4]$, which is an SO
code. Thus, $h([31s+9,5,16s+4])\geq 3$ and the lemma
holds.\end{proof}

 In the rest of this section, we will use some results of
Section 2 to calculate $h(\mathcal{C})$ for each  code $\mathcal{C}$
$=[n,5,d_{a}]$. From now on, we fix $k=5$ and $N=31$, let
$L=(l_{1},l_{2},\cdots,l_{N})$ be a  defining  vector of a given
$[n,5,d_{a}]$ code, and let  $l_{max}=max_{1\leq i\leq N}\{ l_{i}
\}$, $l_{min}=min_{1\leq i\leq N}\{  l_{i} \}$.  For clarity, the
following example is given to show the process of finding $L$ and
calculating $h([n,5,d_{a}])$.

\begin{example}  Let $s\geq 1$, $\mathcal{C}$
$=[31s+13,5,16s+6]$ be an optimal code. One can check $\sigma=
2^{4}+6$ and $s-1\leq l_{i}\leq s+1$ for defining vector
$L=(l_{1},l_{2},\cdots,l_{N})$ of $\mathcal{C}$. According to [16],
there is no $[13,5,6]$ code, thus $l_{max}=s+1$ and $l_{min}=s-1$.
Hence, $L=(s-1){\bf 1_{N}}+L'$, where $L'$ is a defining vector  of
a $[44,5, 22]$ code. We can assume the type of $L'$ is $]]
(0)_{a}\mid (1)_{b}\mid (2)_{c}]]$, where $a\geq 1$, $a+b+c=31$ and
$b+2c=44$. From the system of linear equations ($\star$), one can
obtain

$$(L')^{T}=\frac{1}{16}[12\cdot
{\bf1^{T}_{2^{k}-1}}-2Q_{k}\Lambda^{T}]\ \ \ \ (\star').$$

By solving the  system of linear equations ($\star'$), we get all
possible  $L'$ and $L$. There are totally 4805 solutions of
($\star'$), these $(L')$'s can be divided into
 two groups, one group has 3720 solutions, and the other has 1085 solutions.
 Using Magma [23], one can check that all the $(L')$'s in the same group
 give  equivalent codes. Hence there are altogether two inequivalent $[31s+13,5,16s+6]$
 codes. Much more details of $h([31s+13,5,16s+6])$ and  weight enumerators of
inequivalent $[31s+13,5,16s+6]$ codes are given in  the following
lemma. \end{example}

\begin{lemma} If $s\geq 1$, then a $[31s+13,5,16s+6]$ and a
$[31s+17,5,16s+8]$ codes all have $h\geq 3$.\end{lemma}

\begin{proof} {\it Case 1.} Let $n=31s+13$,  $d=16s+6$, $\mathcal{C}$
$=[n,5,d]$ and $L=(l_{1},l_{2},\cdots,l_{N})$ be a defining  vector
of $\mathcal{C}$. Then one can check $\sigma= 2^{4}+6$ and $s-1\leq
l_{i}\leq s+1$ for $1\leq i\leq N$.
 Since there is no $[13,5,6]$ code, thus the defining vector $L$  may have $l_{max}=s+1$
and $l_{min}=s-1$, which implies  $L=(s-1){\bf 1_{N}}+L'$, where
$L'$ is a defining vector  of  a  $[44,5, 22]$ code. In this case
$\mathcal{C}$ is  the juxtaposition of $(s-1)\mathcal{S}$$_{5}$ and
a $[44,5, 22]$ code. Suppose $L$ is of type $]] (s-1)_{a}\mid
(s)_{b} \mid(s+1)_{c}]]$ with $a\geq 1$. By solving the  system of
linear equations ($\star$), one can obtain that $L'$ is one of the
following two types $]] (0)_{a}\mid (1)_{b}\mid (2)_{c}]]$:

 $L'_{1}$: $]] (0)_{1}\mid (1)_{16} \mid(2)_{14}]]$;
 $L'_{2}$: $]] (0)_{3}\mid (1)_{12}\mid (2)_{16}]]$.

There are 3720  solutions $(L')$'s that  are of type $L'_{1}$, all
these 3720 defining  vectors give equivalent $[44,5, 22]$ codes,
they are equivalent  to a code with defining vector $L'_{1,1}$,
where
 $L'_{1,1}$ $=(11111 0 11111 1111 22222 22222 212122)$.
One can check the corresponding code $\mathcal{C}$ has
$h=h(\mathcal{C})=3$ and weight enumerator
 $1+23y^{16s+6}  +7y^{16s+8}+y^{16s+14}.$

 There are  1085  solutions $(L')$'s that  are of type $L'_{2}$, all these
1085 defining  vectors give equivalent $[44,5, 22]$ codes, they are
equivalent  to a code with defining vector $L'_{2,1}$, where
 $L'_{2,1}$ $=( 11111 0 111101011  22222 22222 222222)$.\\
One can check  the corresponding code $\mathcal{C}$ has
$h=h(\mathcal{C})=3$ and weight enumerator  $1+24y^{16s+6}
+6y^{16s+8}+y^{16s+16}.$

Summarizing previous discussions, we have $h([31s+13,5,16s+6])=3$
and
 $\mathcal{C}$ is not an LCD code.

{\it Case 2.} Let $n=31s+17$ and $d=16s+8$, $\mathcal{D}$ $=[n,5,d]$
and $L=(l_{1},l_{2},\cdots,l_{N})$ be a defining vector of
$\mathcal{D}$. It is easy to check $\sigma= 2^{4}+8$ and $s-1\leq
l_{i}\leq s+2$ for $1\leq i\leq N$. Thus the defining vector $L$ of
$\mathcal{D}$ may be one of the following types:

 (1) $l_{max}=s+2$; (2) $l_{max}=s+1$ and $l_{min}=s$;

 (3) $l_{max}=s+1$ and $l_{min}=s-1$.

If $l_{max}=s+2$, then $\mathcal{D}$ has a reduced code
$[30s+15,4,16s+8]$ $=[15m,4,8m]$ where $m=2s+1$, which is an SO
code, thus one can deduce that $h(\mathcal{D})\geq 3$.

If $l_{max}=s+1$ and $l_{min}=s$, then $L=s{\bf 1_{N}}+L_{0}$, where
$L_{0}$ is a defining  vector  of  a  projective $[17,5, 8]$ code.
In this case $\mathcal{D}$ is  the juxtaposition of
$s\mathcal{S}$$_{5}$ and a  projective $[17,5, 8]$ code. According
to Ref. [21], an  $[17,5, 8]$ code is unique and its $h=4$.

If $l_{max}=s+1$ and $l_{min}=s-1$, then $L=(s-1){\bf 1_{N}}+L'$,
where $L'$ is a defining vector  of  a $[48,5, 24]$ code.  In this
case $\mathcal{D}$ is the juxtaposition of $(s-1)\mathcal{S}$$_{5}$
and a $[48,5, 24]$ code.  Suppose $L$ is of type $]] (s-1)_{a} \mid
(s)_{b} \mid (s+1)_{c}]]$ with $a\geq 1$. By solving the  system of
linear equations ($\star$), we obtain the following types   $]]
(0)_{a}\mid (1)_{b}\mid (2)_{c}]]$ of $L'$:

 $L'_{1}$: $]] (0)_{1}\mid (1)_{12}
\mid(2)_{18}]]$;
 $L'_{2}$: $]] (0)_{3}\mid (1)_{8}\mid (2)_{20}]]$;
$L'_{3}$: $]] (0)_{7}\mid (1)_{0}\mid (2)_{24}]]$.

There are altogether two classes of inequivalent  $[48,5, 24]$ codes
with defining vector of type $]](0)_{1}\mid (1)_{12}
\mid(2)_{18}]]$. Denote their defining vectors as $L'_{1,i}$
($i=1,2$), respectively. Then the corresponding codes $\mathcal{D}$
have $h$
and   weight  enumerators  as follows:\\
$L'_{1,1}$$=(220111121121221 2222222221122112 )$,
$h=5$, $1+24y^{16s+8}  +6y^{16s+12};$\\
$L'_{1,2 }$$=(22021122 1121221 22222 11221121221 )$, $h=3$,
$1+24y^{16s+8}  +8y^{16s+10}+y^{16s+16}.$

There are a class of   $[48,5, 24]$ code with defining vector of
type $]](0)_{3}\mid (1)_{8}\mid (2)_{20}]]$ and a class of   $[48,5,
24]$ code with defining vector of type $]] (0)_{7}\mid (1)_{0}\mid
(2)_{24}]]$, respectively. Denote their defining vector as $L'_{j}$
($j=3,4$). Then the corresponding codes $\mathcal{D}$ have $h$ and
  weight  enumerators as follows:\\
 $L'_{3 }$$=(220200 221121221
222222222  1121221 )$, $h=5$, $1+26y^{16s+8}
+4y^{16s+12}+y^{16s+16};$\\
$L'_{4 }$$=(220200   220020220   22222222  22222222 )$, $h=5$,
$1+28y^{16s+8} +3y^{16s+16}.$

 Summarizing previous discussions, we have $h([31s+17,5,16s+8])=3$.
\end{proof}

From the previous two lemmas and Lemma \ref{lem3}, one can derive
the following conclusion.

\begin{lemma}  The codes
$[31s+8,5,16s+3]$, $[31s+12,5,16s+5]$ and $[31s+16,5,16s+7]$ all
have $h([31s+t,5,d_{a}])\geq 2$, hence they are not LCD
codes.\end{lemma}
 Combining with known results on  $[n,5]$ LCD codes of lengths
 $n=8,9,12,$ $13, 16,33$, we can obtain that
 $[31s+t,5,16s+d_{t}]$ are optimal LCD codes, where
 $d_{t}=-1,2,3,4,5,6$ for $t=2,8,9, 12,13,16$, respectively.

Thus  Theorem \ref{the1.1} holds for the cases of  $t=2,8, 12,16$.

 \subsection{\bf  $h([31s+10,5, 16s+4]) \geq 1$}

 In this subsection,  let $n=31s+10$ and $d=16s+4$, $\mathcal{C}$ $=[n,5,d]$
 and $L=(l_{1},l_{2},\cdots,l_{N})$ be a
defining vector of $\mathcal{C}$. It is easy to check for this code,
$\sigma=2\times 2^{4}+4$ and $s-2\leq l_{i}\leq s+2$ for $1\leq
i\leq N$. Thus the defining vector $L$ of $\mathcal{C}$ may be one
of the following types:

 (1) $l_{max}=s+2$; (2) $l_{max}=s+1$ and $l_{min}=s$;

 (3) $l_{max}=s+1$ and $l_{min}=s-1$; (4) $l_{max}=s+1$ and $l_{min}=s-2$.

If $l_{max}=s+2$, then $\mathcal{C}$ has a reduced code
$[30s+8,4,16s+4]$, which is an SO code. Thus, in this case one can
deduce that $h(\mathcal{C})\geq 3$ and  $\mathcal{C}$ is not an LCD
code.

If $l_{max}=s+1$ and $l_{min}=s$, then $L=s{\bf 1_{N}}+L_{0}$, where
$L_{0}$ is a defining  vector  of  a  projective $[10,5, 4]$ code.
In this case $\mathcal{C}$ is  the juxtaposition of
$s\mathcal{S}$$_{5}$ and a $[10,5, 4]$ code. According to Ref. [10],
a $[10,5, 4]$ code is not an LCD code, hence $\mathcal{C}$ is not an
LCD code either.

For verifying the cases (3) and (4),  two additional lemmas to
determine $h(\mathcal{C})$ are provided as follows.

\begin{lemma} If the defining vector $L=(l_{1},l_{2},\cdots,l_{N})$
of $\mathcal{C}$ satisfies  $l_{max}=s+1$ and $l_{min}=s-1$, then
$h(\mathcal{C})\geq 1$ and  $\mathcal{C}$ is not an LCD
code.\end{lemma}

 \begin{proof} If $l_{max}=s+1$ and $l_{min}=s-1$,
then $s\geq
 1$ and $L=(s-1){\bf 1_{N}}+L'$, where
$L'$ is a define  vector  of  a  $[41,5, 20]$ code.
 In this case $\mathcal{C}$ is  the juxtaposition of $(s-1)\mathcal{S}$$_{5}$ and
a $[41,5, 20]$ code. Suppose $L$ is of type $]] (s-1)_{a}\mid
(s)_{b}\mid (s+1)_{c}]]$ with $a\geq 1$. By solving the  system of
linear equations ($\star$), we obtain the following types of $L'$:

$L'_{1}$: $]] (0)_{1}\mid (1)_{19} \mid(2)_{11}]]$;
 $L'_{2}$: $]] (0)_{2}\mid (1)_{17}\mid (2)_{12}]]$;

$L'_{3}$: $]] (0)_{3}\mid (1)_{15}\mid (2)_{13}]]$;
 $L'_{4}$: $]] (0)_{4}\mid (1)_{13}\mid (2)_{14}]]$;

$L'_{5}$: $]] (0)_{5}\mid (1)_{11}\mid (2)_{15}]]$;
 $L'_{6}$: $]](0)_{6}\mid (1)_{9}\mid (2)_{16}]]$;

$L'_{7}$: $]] (0)_{7}\mid (1)_{7}\mid (2)_{17}]]$.

There are nineteen classes of inequivalent  $[41,5, 20]$ codes with
defining vectors  of the above seven types, all these codes have
$h\geq 1$, hence $h([31s+10,5,16s+4])\geq 1$ when  $L$ satisfying
$l_{max}=s+1$ and $l_{min}=s-1$. For  the defining vectors
$L'_{i,j}$ of these inequivalent  $[41,5, 20]$ codes,
$h(\mathcal{C})$ and weight enumerators of their corresponding
$[31s+10,5,16s+4]$ codes, one can refer to Table \ref{tab2}.

\begin{table}[h!]
\begin{center}
\begin{minipage}{\textwidth}
\caption{19 inequivalent $[31s+10,5, 16s+4]$ codes }\label{tab2}
\begin{tabular*}{\textwidth}{@{\extracolsep{\fill}}lllllll@{\extracolsep{\fill}}}
\toprule
 \multicolumn{3}{c}{Type of defining vector of $L'$: $]](0)_{1}\mid(1)_{19}\mid(2)_{11}]]$} \\
\hline
defining  vector&h&weight  enumerator of $\mathcal{C}$\\
\hline
     (2212121201212112211111121111112)& 3&$1+18y^{16s+4}  +8y^{16s+6}+5y^{16s+8}$  \\
     (2212112201212112211111121111121)&  1& $1+17y^{16s+4}  +11y^{16s+6}+2y^{16s+8}+y^{16s+10}$ \\
     (2212111201212112211112121111112)&  4&$1+12y^{16s+4}  +14y^{16s+5}+3y^{16s+8}+2y^{16s+9}$ \\
\bottomrule
 \multicolumn{3}{c}{Type of defining vector: $]] (0)_{2}\mid (1)_{17} \mid(2)_{12}]]$} \\
\hline
  (0111111222222211122222101111111)&   1& $1+17y^{16s+4}  +12y^{16s+6}+y^{16s+8}+y^{16s+12}$    \\
   (2202112122021121111122111111221)&   4&$1+11y^{16s+4}  +16y^{16s+5}+3y^{16s+8}+y^{16s+12}$   \\
   (2222111201212112210112121111112)&   3&$1+19y^{16s+4}  +7y^{16s+6}+4y^{16s+8}+y^{16s+10}$  \\
    (2222111201212112210111221111121)&  1&$1+18y^{16s+4}  +10y^{16s+6}+y^{16s+8}+2y^{16s+10}$   \\
\bottomrule
 \multicolumn{3}{c}{Type of defining vector: $]] (0)_{3}\mid (1)_{15} \mid(2)_{13}]]$ } \\
\hline
     (2222021201212112210121121111112)& 5&$1+22y^{16s+4} +9y^{16s+8}$  \\
   (2122211202121111021221102122111)&   1&$1+19y^{16s+4}+9y^{16s+6} +3y^{16s+10}$  \\
    (2202112200212112221111121121121)&  3&$1+19y^{16s+4}+8y^{16s+6}+3y^{16s+8} +y^{16s+12}$ \\
    (0111111212222221122222200111111)& 4& $1+12y^{16s+4}+15y^{16s+5} +3y^{16s+8}+y^{16s+13}$ \\
  (2202112200212212221111121121111)&  4&  $1\!\!+\!\!13y^{16s\!+\!4}\!\!+\!\!14y^{16s\!+\!5}\!\!+\!\!y^{16s\!+\!8}\!\!+\!\!2y^{16s\!+\!9}\!\!+\!\!y^{16s\!+\!12}$\\
\bottomrule
  \multicolumn{3}{c}{Type of defining vector: $]] (0)_{4}\mid (1)_{13} \mid(2)_{14}]]$} \\
\hline
   (2021212202121211011212102121212)&  1& $1+18y^{16s+4} +11y^{16s+6}+y^{16s+8} +y^{16s+14}$  \\
   (2202212200212212221101121121111)&  3& $1\!\!+\!\!20y^{16s\!+\!4}\!\!+\!\!7y^{16s\!+\!6}+2y^{16s\!+\!8}\!\!+\!\!y^{16s\!+\!10}\!\!+\!\!y^{16s\!+\!12}$ \\
\bottomrule
 \multicolumn{3}{c}{Type of defining vector: $]] (0)_{5}\mid (1)_{11} \mid(2)_{15}]]$} \\
\hline
   (2202221201212112221100221021121)&  5& $1+23y^{16s+4} +7y^{16s+8}+y^{16s+12}$  \\
\bottomrule
  \multicolumn{3}{c}{Type of defining vector: $]] (0)_{6}\mid (1)_{9} \mid(2)_{16}]]$ } \\
\hline
 1222201102222110022220101222211)&   1&($1+18y^{16s+4} +12y^{16s+8} +y^{16s+16}$  \\
  (1102222111022221001222210012222)& 4& $1+12y^{16s+4}+16y^{16s+6}+2y^{16s+8} +y^{16s+16}$ \\
\bottomrule
  \multicolumn{3}{c}{Type of defining vector: $]] (0)_{7}\mid (1)_{7} \mid(2)_{17}]]$ }   \\
\hline
 (2202002200202212221211221121220)& 3&  $1+20y^{16s+4}+8y^{16s+6} +2y^{16s+8} +y^{16s+16}$  \\
(2202002200202202221211221121221)&  4& $1+14y^{16s+4}+14y^{16s+5}+2y^{16s+9} +y^{16s+16}$\\
\bottomrule
\end{tabular*}

\end{minipage}
\end{center}
\end{table}

\end{proof}

\begin{lemma} If the defining vector
$L=(l_{1},l_{2},\cdots,l_{N})$ of $\mathcal{C}$ satisfies
$l_{max}=s+1$ and $l_{min}=s-2$, then $h(\mathcal{C})\geq 3$ and
$\mathcal{C}$ is not an LCD code.\end{lemma}

\begin{proof} If
$l_{max}=s+1$ and $l_{min}=s-2$, then $s\geq
 2$ and $L=(s-2){\bf 1_{N}}+L''$, where
$L''$ is a defining vector  of  a   $[72,5, 36]$ code.
 In this case $\mathcal{C}$ is  the juxtaposition of $(s-2)\mathcal{S}$$_{5}$ and
a $[72,5, 36]$ code. Suppose $L$ is of type $]] (s-2)_{a}\mid
(s-1)_{b}\mid (s)_{c}\mid (s+1)_{d}]]$ with $a\geq 1$. By solving
system of linear equations ($\star$), we obtain the following six
types of $L''$:

$L''_{1,0}$: $]] (0)_{1}\mid (1)_{0}\mid(2)_{18} \mid(3)_{12}]]$;
 $L''_{1,2}$: $]] (0)_{1}\mid (1)_{2}\mid (2)_{14}\mid (3)_{14}]]$;

$L''_{1,4}$: $]] (0)_{1}\mid (1)_{4}\mid (2)_{10}\mid (3)_{16}]]$;
 $L''_{1,6}$: $]](0)_{1}\mid (1)_{6}\mid (2)_{6}\mid (3)_{18}]]$;

$L''_{3,4}$: $]] (0)_{3}\mid (1)_{4}\mid (2)_{4}\mid (3)_{20}]]$;
 $L''_{7,0}$: $]](0)_{7}\mid (1)_{0}\mid (2)_{0}\mid (3)_{24}]]$.

There are thirteen classes of inequivalent  $[72,5, 36]$ codes with
defining vectors of the above types, seven  classes have $h=5$ and
six  classes have $h=3$, thus all these codes have $h\geq 3$, hence
$h([31s+10,5,16s+4])\geq 3$ when  $L$ satisfying $l_{max}=s+1$ and
$l_{min}=s-2$. For details of the defining vectors $L''_{i,j}$ of
these inequivalent  $[72,5, 36]$ codes, $h(\mathcal{C})$ and weight
enumerators of their corresponding $[31s+10,5,16s+4]$ codes,  see
Table \ref{tab3}.

\begin{table}[h]
\begin{center}
\begin{minipage}{\textwidth}
\caption{13 inequivalent $[31s+10,5, 16s+4]$ codes
 }\label{tab3}
\begin{tabular*}{\textwidth}{@{\extracolsep{\fill}}lllllll@{\extracolsep{\fill}}}
\toprule
 \multicolumn{3}{c}{Type of defining vector of $L''$: $]] (0)_{1}\mid (1)_{0}\mid(2)_{18} \mid(3)_{12}]]$} \\
\hline
defining  vector&h&weight  enumerator  of $\mathcal{C}$ \\
\hline
   (3323232332222220332323233222222)& 5& $1+22y^{16s+4} +9y^{16s+8}$ \\
   (3323232332222220332323233222222)&  3&$1+20y^{16s+4}+6y^{16s+6} +3y^{16s+8} +2y^{16s+10}$\\
  (3323232332222220332323233222222)&    3&$1+19y^{16s+4}+8y^{16s+6} +3y^{16s+8} +y^{16s+12}$ \\
\bottomrule
 \multicolumn{3}{c}{Type of defining vector:  $]] (0)_{1}\mid (1)_{2}\mid(2)_{14} \mid(3)_{14}]]$} \\
\hline
    (3222203333232332122222133323233)& 5& $1+23y^{16s+4} +7y^{16s+8}+y^{16s+12}$  \\
   (3323213233031232332322223322223)&   3&$1+20y^{16s+4}+7y^{16s+6} +3y^{16s+8}+y^{16s+14}$\\
     (3333222333022232331222313323222)&3& $1\!\!+\!\!21y^{16s\!+\!4}\!\!+\!\!6y^{16s\!+\!6}\!\!+\!\!y^{16s\!+\!8}+2y^{16s\!+\!10}\!\!+\!\!y^{16s\!+\!12}$ \\
\bottomrule
 \multicolumn{3}{c}{Type of defining vector:  $]] (0)_{1}\mid (1)_{4}\mid(2)_{10} \mid(3)_{16}]]$} \\
\hline
   (3333303332121332331312322323222)& 5&  $1+24y^{16s+4} +5y^{16s+8}+2y^{16s+12}$ \\
  (3323203332131232332322323313123) & 3&  $1+20y^{16s+4}+8y^{16s+6} +2y^{16s+8} +y^{16s+16}$\\
\bottomrule
 \multicolumn{3}{c}{Type of defining vector:  $]] (0)_{1}\mid (1)_{6}\mid(2)_{6} \mid(3)_{18}]]$} \\
\hline
     (3333303233131232331312323313123)&5& $1+24y^{16s+4} +6y^{16s+8}+y^{16s+16}$  \\
     (3313103333232332113132133323233)&5& $1+20y^{16s+4}+7y^{16s+6}+3y^{16s+8}+y^{16s+14}$\\
     (3333123333032132331321313323123)&3& $1+22y^{16s+4}+6y^{16s+6} +2y^{16s+10}+y^{16s+16}$ \\
\bottomrule
 \multicolumn{3}{c}{Type of defining vector:   $]] (0)_{3}\mid (1)_{4}\mid(2)_{4} \mid(3)_{20}]]$} \\
\hline
  (3333303333030332331312313323213)&   5&  $1+26y^{16s+4} +2y^{16s+8} +2y^{16s+12}+y^{16s+16}$\\
    \bottomrule
 \multicolumn{3}{c}{Type of define vector: $]] (0)_{7}\mid (1)_{0}\mid(2)_{0} \mid(3)_{24}]]$} \\
\hline
     (3333303333030330333330333303033)& 5&$1+28y^{16s+4} +3y^{16s+16}$\\
   \bottomrule
\end{tabular*}
\end{minipage}
\end{center}
\end{table}

Summarizing the above, we have shown $h([31s+10,5,16s+4])\geq 1$
 for all $s\geq 1$, and there is no $[31s+10,5,16s+4]$ LCD
code.
\end{proof}

\subsection{\bf  $h([31s+14,5, 16s+6])\geq 1$ }

 In this subsection,  let $n=31s+14$, $d=16s+6$, $\mathcal{C}=[n,5,d]$ and $L=(l_{1},l_{2},\cdots,l_{N})$ be a
defining vector of $\mathcal{C}$. It is easy to check for this code,
$\sigma=2\times 2^{4}+6$ and $s-2\leq l_{i}\leq s+2$ for $1\leq
i\leq N$. Thus the defining vector $L$ of $\mathcal{C}$ may have the
following types:

 (1) $l_{max}=s+2$; (2) $l_{max}=s+1$ and $l_{min}=s$;

 (3) $l_{max}=s+1$ and $l_{min}=s-1$; (4) $l_{max}=s+1$ and
 $l_{min}=s-2$.

If $l_{max}=s+2$, then $\mathcal{C}$ has a reduced code
$[30s+12,4,16s+6]$, which is a  code with $h([30s+12,4,16s+6])=2$.
Thus, in this case one can deduce that $h(\mathcal{C})\geq 1$,
$\mathcal{C}$ is not an LCD code.

If $l_{max}=s+1$ and $l_{min}=s$, then $L=s{\bf 1_{N}}+L_{0}$, where
$L_{0}$ is a defining  vector  of  a  projective $[14,5, 6]$ code.
In this case $\mathcal{C}$ is the juxtaposition of
$s\mathcal{S}$$_{5}$ and a $[14,5, 6]$ code. According to
Refs.[10,11], one can know a $[14,5, 6]$ code is not an LCD code.
Hence $\mathcal{C}$ is not an LCD code.

 \begin{lemma} If the defining vector $L=(l_{1},l_{2},\cdots,l_{N})$
of $\mathcal{C}$ satisfies  $l_{max}=s+1$ and $l_{min}=s-1$, then
$h(\mathcal{C})\geq 1$ and $\mathcal{C}$ is not an LCD code.
\end{lemma}

 \begin{proof} If $l_{max}=s+1$ and
$l_{min}=s-1$, then  $L=(s-1){\bf 1_{N}}+L'$, where $L'$ is a
defining vector  of  a   $[45,5, 22]$ code.
 In this case $\mathcal{C}$ is  the juxtaposition of $(s-1)\mathcal{S}$$_{5}$ and
a $[45,5, 22]$ code. Suppose $L$ is of type $]] (s-1)_{a}\mid
(s)_{b}\mid (s+1)_{c}]]$ with $a\geq 1$. By solving the  system of
linear equations ($\star$), we obtain the following types  $]]
(0)_{a}\mid (1)_{b}\mid (2)_{c}]]$ of $L'$:

$L'_{1}$: $]] (0)_{1}\mid (1)_{15} \mid(2)_{15}]]$;
 $L'_{2}$: $]] (0)_{2}\mid (1)_{13}\mid (2)_{16}]]$;

$L'_{3}$: $]] (0)_{3}\mid (1)_{11}\mid (2)_{17}]]$;
 $L'_{4}$: $]] (0)_{4}\mid (1)_{9}\mid (2)_{18}]]$;

$L'_{5}$: $]] (0)_{5}\mid (1)_{7}\mid (2)_{19}]]$;
 $L'_{6}$: $]](0)_{6}\mid (1)_{5}\mid (2)_{20}]]$;

$L'_{7}$: $]] (0)_{7}\mid (1)_{3}\mid (2)_{21}]]$.
\begin{table}[h]
\begin{center}
\begin{minipage}{\textwidth}
\caption{ 21 inequivalent $[31s+14,5, 16s+6]$ codes}\label{tab4}
\begin{tabular*}{\textwidth}{@{\extracolsep{\fill}}lllllll@{\extracolsep{\fill}}}
\toprule
 \multicolumn{3}{c}{Type of defining vector of $L'$: $]] (0)_{1}\mid (1)_{15} \mid(2)_{15}]]$} \\
\hline
defining  vector&h&weight  enumerator of $\mathcal{C}$\\
\hline
  (2212112122121120112122121121221)& 5&  $1+15y^{16s+6} +15y^{16s+8} +y^{16s+14}$ \\
   (2212112211211112220211221121221)& 3& $1+15y^{16s+6} +15y^{16s+8} +y^{16s+14}$\\
   (2211111211112112221221221122022)&  3&  $1+18y^{16s+6} +7y^{16s+8} +6y^{16s+10}$\\
(2111112121222220111111221122222)&1&$1+17y^{16s+6}+10y^{16s+8}\!\!+\!\!3y^{16s+10}\!\!+\!\!y^{16s+12}$\\
(1111111222222221222222201111111)&3&$1+8y^{16s+6}\!\!+\!\!15y^{16s+7}\!\!+\!\!7y^{16s+8}\!\!+\!\!y^{16s+15}$\\
(2212112211211212220211221121211)&2&$1\!\!+\!\!11y^{16s\!+\!6}\!\!+\!\!12y^{16s\!+\!6}\!\!+\!\!3y^{16s\!+\!8}\!\!+\!\!4y^{16s+9}\!\!+\!\!y^{16s\!+\!14}$\\
\bottomrule
 \multicolumn{3}{c}{Type of defining vector:  $]] (0)_{2}\mid (1)_{13} \mid(2)_{16}]]$}\\
\hline
    (1111112122222220011111221222222)& 1& $1+18y^{16s+6}+10y^{16s+8} +2y^{16s+10} +y^{16s+14}$\\
    (1211112122122220012111221122222)& 3& $1+19y^{16s+6}+6y^{16s+8} +5y^{16s+10} +y^{16s+12}$\\
    (1112222111122220111222201112222)& 3& $1+8y^{16s+6}+16y^{16s+7} +6y^{16s+8} +y^{16s+16}$ \\
    (1111122121222220011112221122222)& 1&  $1+18y^{16s+6}+9y^{16s+8} +2y^{16s+10} +2y^{16s+12}$\\
\bottomrule
 \multicolumn{3}{c}{Type of defining vector:  $]] (0)_{3}\mid (1)_{11} \mid(2)_{17}]]$} \\
\hline
   (1112222111122220011222220112222)&  3& $1+18y^{16s+6}+10y^{16s+8} +2y^{16s+10} +y^{16s+14}$\\
   (2111120122222220111112021222222)& 1&  $1+18y^{16s+6}+7y^{16s+8} +4y^{16s+10} +2y^{16s+14}$\\
   (1111122122122220002112221122222)& 1&  $1+20y^{16s+6}+5y^{16s+8} +4y^{16s+10} +2y^{16s+12}$ \\
    (2111222122022210111122221202221)&3&  $1+20y^{16s+6}+8y^{16s+8} +y^{16s+10} +2y^{16s+12}$\\
    (2212102211212212220201221121221)&2&   $1\!\!+\!\!14y^{16s\!+\!6}\!\!+\!\!10y^{16s\!+\!7}\!\!+\!\!2y^{16s\!+\!8}\!\!+\!\!4y^{16s\!+\!9}
\!\!+\!\!y^{16s+16}$\\
\bottomrule
 \multicolumn{3}{c}{Type of defining vector:   $]] (0)_{4}\mid (1)_{9} \mid(2)_{18}]]$} \\
\hline
   (1112222121022220011222221102222)&  1& $1+18y^{16s+6}+10y^{16s+8} +2y^{16s+10} +y^{16s+16}$\\
   (1211022122122220012102221122222)& 3& $1+21y^{16s+6}+4y^{16s+8} +3y^{16s+10} +3y^{16s+12}$\\
   (1101222122122220001122221122222)&  3& $1\!\!+\!\!20y^{16s\!+\!6}\!\!+\!\!6y^{16s\!+\!8}\!\!+\!\!3y^{16s\!+\!10}
\!\!+\!\!y^{16s\!+\!12}\!\!+\!\!y^{16s\!+\!14}$\\
\bottomrule
 \multicolumn{3}{c}{Type of defining vector: $]] (0)_{5}\mid (1)_{7} \mid(2)_{19}]]$} \\
\hline
    (1112222112022220002222220112222)&  3&$1+20y^{16s+6}+6y^{16s+8} +4y^{16s+10} +y^{16s+16}$\\
\bottomrule
 \multicolumn{3}{c}{Type of defining vector: $]]
(0)_{6}\mid (1)_{5} \mid(2)_{20}]]$} \\
\hline
   (1122222120022220012222221002222)&  1& $1\!+\!20y^{16s+6}\!+\!8y^{16s+8}\!+\!2y^{16s+12}\!+\!y^{16s+16}$\\
\bottomrule
 \multicolumn{3}{c}{Type of defining vector:  $]] (0)_{7}\mid (1)_{3}
\mid(2)_{21}]]$} \\
\hline
   (2222220122002220122222021200222)&  3& $1+21y^{16s+6}+7y^{16s+8}+3y^{16s+14}$\\
\bottomrule
\end{tabular*}
\end{minipage}
\end{center}
\end{table}

There are twenty one classes of inequivalent  $[45,5, 22]$ codes
with defining vector of the above seven types. And all these codes
have $h\geq 1$, hence $h([31s+14,5,16s+6])\geq 1$ when $L$
satisfying $l_{max}=s+1$ and $l_{max}=s-1$. For details of the
defining vectors $L'_{i,j}$ of these inequivalent  $[45,5, 22]$
codes, $h(\mathcal{C})$ and weight  enumerators of their
corresponding $[31s+14,5,16s+6]$ codes, one can refer to Table
\ref{tab4}.
\end{proof}

\begin{lemma} If the defining vector
$L=(l_{1},l_{2},\cdots,l_{N})$ of $\mathcal{C}$ satisfies
$l_{max}=s+1$ and $l_{min}=s-2$, then $h(\mathcal{C})\geq 3$ and
$\mathcal{C}$ is not an LCD code. \end{lemma}

\begin{proof} If $l_{max}=s+1$ and $l_{max}=s-2$, then $s\geq
 2$ and $L=(s-2){\bf 1_{N}}+L''$,
where $L''$ is a defining  vector  of  a  $[76,5,38]$ code.
 In this case $\mathcal{C}$ is  the juxtaposition of $(s-2)\mathcal{S}$$_{5}$ and
a $[76,5,38]$ code. Suppose $L$ is of type $]]
(s-2)_{a}\mid(s-1)_{b}\mid (s)_{c}\mid (s+1)_{d}]]$ with $a\geq 1$.
By solving the  system of linear equations ($\star$), we obtain the
following types $]] (0)_{a} \mid(1)_{b}\mid (2)_{c}\mid(3)_{d}]]$ of
$L''$:

$L''_{1,0}$: $]] (0)_{1}\mid (1)_{0}\mid(2)_{14} \mid(3)_{16}]]$;
 $L''_{1,2}$: $]] (0)_{1}\mid (2)_{2}\mid (2)_{10}\mid (3)_{18}]]$;

$L''_{1,4}$: $]] (0)_{1}\mid (1)_{4}\mid (2)_{6}\mid (3)_{20}]]$;
$L''_{1,6}$: $]] (0)_{1}\mid (1)_{6}\mid (2)_{2}\mid (3)_{22}]]$;

 $L''_{3,0}$: $]](0)_{3}\mid (1)_{0}\mid (2)_{8}\mid (3)_{20}]]$;
 $L''_{3,4}$: $]](0)_{3}\mid (1)_{4}\mid (2)_{0}\mid (3)_{24}]]$.

There are 10 classes of inequivalent  $[76,5,38]$ codes with the
defining vectors of the above six types. And all these codes have
$h\geq 3$, hence $h([31s+14,5,16s+6])\geq 3$ when  $L$ satisfying
$l_{max}=s+1$ and $l_{min}=s-2$. For defining vector $L''_{i,j}$ of
these inequivalent  $[76,5,38]$ codes, $h(\mathcal{C})$ and their
weight  enumerators of $[31s+14,5,16s+6]$ codes, see Table
\ref{tab5}.

\begin{table}[h]
\begin{center}
\begin{minipage}{\textwidth}
\caption{ 10 inequivalent $[31s+14,5, 16s+6]$ codes}\label{tab5}
\begin{tabular*}{\textwidth}{@{\extracolsep{\fill}}lllllll@{\extracolsep{\fill}}}
\toprule
 \multicolumn{3}{c}{Type of defining vector of $L''$:  $]] (0)_{1}\mid (1)_{0}\mid(2)_{14} \mid(3)_{16}]]$} \\
\hline
defining  vector&h&weight  enumerator  of $\mathcal{C}$\\
\hline
  (3323223233232232332322303323223)& 5& $1+16y^{16s+6} +14y^{16s+8}+y^{16s+16}$\\
  (3323223233322232330322333323222)& 3& $1+19y^{16s+6} +7y^{16s+8} +4y^{16s+10} +y^{16s+14}$\\
  (3323223323322232330322333233222)& 3& $1+20y^{16s+6} +5y^{16s+8}+4y^{16s+10} +2y^{16s+12}$\\
\bottomrule
 \multicolumn{3}{c}{Type of defining vector:  $]] (0)_{1}\mid (1)_{2}\mid(2)_{10} \mid(3)_{18}]]$}\\
\hline
  (3323223233332132330322333323123)&  3&$1+20y^{16s+6}+6y^{16s+8}+4y^{16s+10}+y^{16s+16}$\\
   (3323313323332222330313333232223)& 3& $1+22y^{16s+6}+3y^{16s+8}+2y^{16s+10}+4y^{16s+12}$\\
 (3323213233323232330312333323232)&  3& $1\!\!+\!\!21y^{16s\!+\!6}\!\!+\!\!5y^{16s\!+\!8}\!\!+\!\!2y^{16s\!+\!10}
\!\!+\!\!2y^{16s\!+\!12}\!\!+\!\!y^{16s\!+\!14}$\\
\bottomrule
 \multicolumn{3}{c}{Type of defining vector: $]] (0)_{1}\mid (1)_{4}\mid(2)_{2} \mid(3)_{20}]]$} \\
\hline
 (3323113233332332330311333323323)&   3&$1\!\!+\!\!22y^{16s+6}\!\!+\!\!4y^{16s\!+\!8}\!\!+\!\!2y^{16s\!+\!10}
\!\!+\!\!2y^{16s\!+\!12}\!\!+\!\!y^{16s\!+\!16}$\\
\bottomrule
 \multicolumn{3}{c}{Type of defining vector: $]] (0)_{1}\mid (1)_{6}\mid(2)_{2} \mid(3)_{22}]]$} \\
\hline
  (3323203333131333331313313333313)&  3&$1+22y^{16s+6}+6y^{16s+8}  +2y^{16s+14}+y^{16s+16}$\\
\bottomrule
 \multicolumn{3}{c}{Type of defining vector: $]] (0)_{3}\mid (1)_{0}\mid(2)_{8} \mid(3)_{20}]]$} \\
\hline
  (3323203233333232330302333323233)& 3& $1 +24y^{16s+6}+2y^{16s+8}  +4y^{16s+12}+y^{16s+16}$\\
\bottomrule
 \multicolumn{3}{c}{Type of defining vector: $]] (0)_{3}\mid (1)_{4}\mid(2)_{0} \mid(3)_{24}]]$} \\
\hline
  (3333303333131333330303313333313)& 3&
  $1+24y^{16s+6}+4y^{16s+8}+3y^{16s+16}$\\
\hline
\end{tabular*}
\end{minipage}
\end{center}
\end{table}

Summarizing the above, we have shown $h([31s+14,5,16s+6])\geq 1$
holds for all $s\geq 1$, and there is no $[31s+14,5,16s+6]$ LCD
code.
\end{proof}
\subsection{\bf  $h([31s+18,5, 16s+8])\geq 1$ }

 In this subsection, we let $n=31s+18$ and $d=16s+8$, $\mathcal{C}$ $=[n,5,d]$,
  and $L=(l_{1},l_{2},\cdots,l_{N})$ be a
defining vector of $\mathcal{C}$.  It is easy to check for this
code, $\sigma=2\times 2^{4}+8$ and $s-2\leq l_{i}\leq s+3$ for
$1\leq i\leq N$. Thus the defining vector $L$ of $\mathcal{C}$ may
have the following types:

 (1) $l_{max}=s+3$; (2) $l_{max}=s+2$;  (3)  $l_{max}=s+1$ and $l_{min}=s$;

(4) $l_{max}=s+1$ and $l_{min}=s-1$; (5) $l_{max}=s+1$ and
 $l_{min}=s-2$.

If $l_{max}=s+3$, then $\mathcal{C}$ has a reduced code
$[30s+15,4,16s+8]$, which is an SO code, thus one can deduce that
$h(\mathcal{C})\geq 3$, $\mathcal{C}$ is not  LCD.

If $l_{max}=s+2$, then $\mathcal{C}$ has a reduced code
$[30s+16,4,16s+8]$$=[15m+1,4,8m]$ for $m=2s+1$, which  is a code
with $h\geq 2$,  then one can deduce that $h(\mathcal{C})\geq 1$ and
$\mathcal{C}$ is not LCD.

If $l_{max}=s+1$ and $l_{min}=s$, then $L=s{\bf 1_{N}}+L_{0}$, where
$L_{0}$ is a defining  vector  of  a  projective $[18,5, 8]$ code.
In this case $\mathcal{C}$ is  the juxtaposition of
$s\mathcal{S}$$_{5}$ and an $[18,5, 8]$ code. According to [10,11],
an $[18,5, 8]$ code is not  LCD and $h([18,5, 8])\geq 1$, hence
$h(\mathcal{C})\geq 1$ and $\mathcal{C}$ is not LCD.

For $L$ satisfying (4) or (5), we use two lemmas to check
$h(\mathcal{C})\geq 1$.

\begin{lemma}If the defining vector
$L=(l_{1},l_{2},\cdots,l_{N})$ of $\mathcal{C}$ satisfies
$l_{max}=s+1$ and $l_{min}=s-2$, then
 $h(\mathcal{C})\geq 1$ and $\mathcal{C}$ is not an LCD code. \end{lemma}

\begin{proof} If $l_{max}=s+1$ and $l_{max}=s-1$, then $L=(s-1){\bf
1_{N}}+L'$, where  $L'$ is a defining  vector  of  a  $[49,5, 24]$
code.  In  this case $\mathcal{C}$ is  the juxtaposition of
$(s-1)\mathcal{S}$$_{5}$ and a $[49,5, 24]$ code. By solving the
system of linear equations ($\star$), we obtain the following types
of $L'$:

$L'_{1}$: $]] (0)_{1}\mid (1)_{11} \mid(2)_{19}]]$;
 $L'_{2}$: $]] (0)_{2}\mid (1)_{9}\mid (2)_{20}]]$;
$L'_{3}$: $]] (0)_{3}\mid (1)_{7}\mid (2)_{21}]]$;

 $L'_{4}$: $]] (0)_{4}\mid (1)_{5}\mid (2)_{22}]]$;
$L'_{5}$: $]](0)_{6}\mid (1)_{1}\mid (2)_{24}]]$.

There are fifteen  classes of inequivalent  $[49,5, 24]$ codes with
defining vector of the above five types, all these codes have $h\geq
1$, hence $h([31s+18,5,16s+8])\geq 1$ when  $L$ satisfies
$l_{max}=s+1$ and $l_{max}=s-1$. For details of the defining vectors
$L'_{i,j}$ of these inequivalent  $[49,5, 24]$ codes,
$h(\mathcal{C})$ and their weight  enumerators of $[31s+18,5,16s+8]$
codes, see Table \ref{tab6}.

\begin{table}[h]
\begin{center}
\begin{minipage}{\textwidth}
\caption{ 15 inequivalent $[31s+18,5, 16s+8]$ codes}\label{tab6}
\begin{tabular*}{\textwidth}{@{\extracolsep{\fill}}lllllll@{\extracolsep{\fill}}}
\toprule
 \multicolumn{3}{c}{Type of defining vector of  $L'$: $]] (0)_{1}\mid (1)_{11} \mid(2)_{19}]]$} \\
\hline
defining  vector&h&weight   enumerator   of $\mathcal{C}$\\
\hline
  (2212112122121122221210222212122)&  3&$1+15y^{16s+8} +16y^{16s+10} +y^{16s+16}$\\
  (2212121211122122221210222221122)&  3&$1+17y^{16s+8} +8y^{16s+10} +6y^{16s+12}$\\
  (2212121211221122221210222211222)&  1&$1+16y^{16s+8} +11y^{16s+10} +3y^{16s+12}+y^{16s+14}$\\
  (2212122211122122221210222221121)&  1&$1+11y^{16s+8} +14y^{16s+9} +4y^{16s+12}+2y^{16s+13}$\\
  (2212112122121222221211222212102)&  2&$1\!\!+\!\!10y^{16s\!+\!8}\!\!+\!\!12y^{16s\!+\!9}\!\!+\!\!4y^{16s\!+\!10}\!\!+\!\!4y^{16s\!+\!11}\!\!+\!\!y^{16s\!+\!16}$\\
\bottomrule
 \multicolumn{3}{c}{Type of defining vector:$]] (0)_{2}\mid (1)_{9} \mid(2)_{20}]]$} \\
\hline
  (2212112122122222221210222212012)&  1&$1+16y^{16s+8} +12y^{16s+10}+2y^{16s+12} +y^{16s+16}$\\
  (2212122122121222221210222212102)&  4&$1+10y^{16s+8} +16y^{16s+9} +4y^{16s+12}+y^{16s+16}$\\
  (2212212212122212221200222121122)&  3&$1+18y^{16s+8} +7y^{16s+10} +5y^{16s+12}+y^{16s+14}$\\
  (2212212211122222221200222221121)&  1&$1+17y^{16s+8} +10y^{16s+10} +2y^{16s+12}+2y^{16s+14}$\\
\bottomrule
 \multicolumn{3}{c}{Type of defining vector:$]] (0)_{3}\mid (1)_{7} \mid(2)_{21}]]$} \\
\hline
 2212112211212012220222222222022)&   5&($1+21y^{16s+8} +10y^{16s+12}$\\
  (2222201212222122220202222111122)&  3&$1+18y^{16s+8} +8y^{16s+10} +4y^{16s+12}+y^{16s+16}$\\
  (2222202212122212220201222121122)&  1&$1+10y^{16s+8} +9y^{16s+10} +y^{16s+12}+3y^{16s+14}$\\
  (2222202212121212220202222121212)&  4& $1\!\!+\!\!12y^{16s\!+\!8}\!\!+\!\!14y^{16s\!+\!9}\!\!+\!\!2y^{16s\!+\!12}\!\!+\!\!2y^{16s\!+\!13}\!\!+\!\!y^{16s\!+\!16}$\\
\bottomrule
 \multicolumn{3}{c}{Type of defining vector:$]] (0)_{4}\mid (1)_{5} \mid(2)_{22}]]$} \\
\hline
  (2212212122222022221200222222012)&  1&$1+18y^{16s+8} +10y^{16s+10} +2y^{16s+14}+y^{16s+16}$\\
\bottomrule
 \multicolumn{3}{c}{Type of defining vector:$]] (0)_{6}\mid
(1)_{1} \mid(2)_{24}]]$} \\
\hline
  (2222202122222022220202222202022)&  4& $1+12y^{16s+8} +16y^{16s+9} +3y^{16s+16}$\\
\bottomrule
\end{tabular*}
\end{minipage}
\end{center}
\end{table}
\end{proof}

\begin{lemma} If the defining vector
$L=(l_{1},l_{2},\cdots,l_{N})$ of $\mathcal{C}$ satisfies
$l_{max}=s+1$ and $l_{min}=s-2$, then $\mathcal{C}$ is not an LCD
code. \end{lemma}

\begin{proof} If $l_{max}=s+1$ and $l_{max}=s-2$, then $s\geq
 2$ and $L=(s-2){\bf 1_{N}}+L''$,
where $L''$ is a defining  vector  of  an  $[80,5,40]$ code.
 In this case $\mathcal{C}$ is  the juxtaposition of $(s-2)\mathcal{S}$$_{5}$ and
an $[80,5,40]$ code. Suppose $L$ is of type $]]
(s-2)_{a}\mid(s-1)_{b} \mid(s)_{c}\mid (s+1)_{d}]]$ with $a\geq 1$.
By solving the  system of linear equations ($\star$), we obtain the
following types  of $L''$:

$L''_{1,0}$: $]] (0)_{1}\mid (1)_{0}\mid(2)_{10} \mid(3)_{20}]]$;
 $L''_{1,2}$: $]] (0)_{1}\mid (2)_{2}\mid (2)_{6}\mid (3)_{22}]]$;

$L''_{1,4}$: $]] (0)_{1}\mid (1)_{4}\mid (2)_{2}\mid (3)_{24}]]$;
 $L''_{3,0}$: $]](0)_{3}\mid (1)_{0}\mid (2)_{4}\mid (3)_{24}]]$.

There are seven  classes of inequivalent  $[80,5,40]$ codes with
defining vector of the above four types, all these codes have $h\geq
3$, hence $h([31s+18,5,16s+8])\geq 3$ when  $L$ satisfying
$l_{max}=s+1$ and $l_{max}=s-2$. For details of the defining vectors
$L''_{i,j}$ of these inequivalent  $[80,5,40]$ codes, and
$h(\mathcal{C})$ and weight  enumerators of their corresponding
 $[31s+18,5,16s+8]$ codes, see  Table \ref{tab7}.

\begin{table}[h]
\begin{center}
\begin{minipage}{\textwidth}
\caption{ 7 inequivalent $[31s+18,5, 16s+8]$ codes }\label{tab7}
\begin{tabular*}{\textwidth}{@{\extracolsep{\fill}}lllllll@{\extracolsep{\fill}}}
\toprule
 \multicolumn{3}{c}{Type of defining vector $L''$:  $]] (0)_{1}\mid (1)_{0}\mid(2)_{10} \mid(3)_{20}]]$} \\
\hline
defining  vector&h&weight  enumerator   of $\mathcal{C}$\\
\hline
  (3333233323332223330332333232223)& 5& $1+21y^{16s+8} +10y^{16s+12}$\\
  (3332332323332233330233233233223)& 3& $1+18y^{16s+8} +8y^{16s+10} +4y^{16s+12} +y^{16s+16}$\\
  (3323233322233333332320333332232)& 3& $1+19y^{16s+6} +6y^{16s+10} +4y^{16s+12}+2y^{16s+14}$\\
\bottomrule
 \multicolumn{3}{c}{Type of defining vector:   $]] (0)_{1}\mid (1)_{2}\mid(2)_{6} \mid(3)_{22}]]$} \\
\hline
  (3323203322323323331313333333333)& 5&  $1+22y^{16s+8} +8y^{16s+12}+y^{16s+16}$\\
  (3333303323232323331313333232333)& 3& $1\!\!+\!\!20y^{16s\!+\!8}\!\!+\!\!6y^{16s\!+\!10}\!\!+\!\!2y^{16s\!+\!12}\!\!+\!\!2y^{16s\!+\!14}
\!\!+\!\!y^{16s\!+\!16}$\\
\bottomrule
 \multicolumn{3}{c}{Type of defining vector:    $]] (0)_{1}\mid (1)_{4}\mid(2)_{2} \mid(3)_{24}]]$} \\
\hline
   (33333313233333133330133333233133)& 3&  $1+21y^{16s+8} +8y^{16s+10}+3y^{16s+16}$\\

\bottomrule
 \multicolumn{3}{c}{Type of defining vector:  $]] (0)_{3}\mid (1)_{0}\mid(2)_{4} \mid(3)_{24}]]$} \\
\hline
   (3333303323333323330303333232333)& 5&  $1+24y^{16s+8}+4y^{16s+12} +3y^{16s+16}$\\
\bottomrule
\end{tabular*}
\end{minipage}
\end{center}
\end{table}

Summarizing the above, we have shown $h([31s+18,5,16s+8])\geq 1$ for
all $s\geq 1$ and there is no $[31s+18,5,16s+8]$ LCD code.

\end{proof}
\section{Conclusion}

Combining with known results on optimal LCD codes, the minimum
distances of all binary optimal LCD codes of dimension 5 have been
wiped out in this manuscript. More precisely, we have determined the
minimum distances of optimal $[n,5]$ LCD codes  with $n=31s+t\geq
14$ and $t \in\{2,8,10,12,14,16,18\}$, which haven't been
systematically investigated in the literature. By the  methods  of
reduced codes, classifying optimal linear codes and calculating the
hull dimension of $\mathcal{C}$, one may further study the
classification of optimal linear codes and determine the minimum
distances of optimal LCD codes with higher dimensions.

\section*{Acknowledgements}
This work is supported by Natural Science Foundation of Shaanxi
Province under Grant Nos.2021JQ-335, 2022JQ-046, 2023-JC-YB-003,
2023-JC-QN-0033, National Natural Science Foundation of China under
Grant No.U21A20428.

\section*{Statements and Declarations}

We declare that we have no known competing financial interests or
personal relationships that could have appeared to influence the
work reported in this paper.

\end{document}